\documentclass[american,aps,prxquantum,reprint,superscriptaddress, longbibliography,floatfix,nofootinbib]{revtex4-1}
\usepackage[T1]{fontenc}
\usepackage{inputenc}
\usepackage{color}
\usepackage{babel}
\usepackage{physics}
\usepackage{mathtools}
\usepackage{bbding}
\usepackage{pifont}
\usepackage{textcomp}
\usepackage{wasysym}
\usepackage{amsthm}
\usepackage{nicefrac}
\usepackage{wasysym}
\usepackage{dsfont}
\usepackage{amsmath}
\usepackage{amssymb}
\usepackage{graphicx}
\usepackage{enumitem}
\definecolor{darkblue}{rgb}{0.0, 0.0, 0.55}

\usepackage[colorlinks=true,
            allcolors=blue]{hyperref}

\hypersetup{
     colorlinks   = true,
     linkcolor    = blue,
     citecolor    = blue,
     urlcolor     = blue,
     }

\usepackage{appendix}
\renewcommand{\selectlanguage}[1]{}
\makeatletter
\@ifundefined{textcolor}{}
{%
	\definecolor{BLACK}{gray}{0}
	\definecolor{WHITE}{gray}{1}
	\definecolor{RED}{rgb}{1,0,0}
	\definecolor{GREEN}{rgb}{0,1,0}
	\definecolor{BLUE}{rgb}{0,0,1}
	\definecolor{CYAN}{cmyk}{1,0,0,0}
	\definecolor{MAGENTA}{cmyk}{0,1,0,0}
	\definecolor{YELLOW}{cmyk}{0,0,1,0}
}
\theoremstyle{plain}

\theoremstyle{plain}

\ifx\proof\undefined
\newenvironment{proof}[1][\protect\proofname]{\par
	\normalfont\topsep6\p@\@plus6\p@\relax
	\trivlist
	\itemindent\parindent
	\item[\hskip\labelsep
	\scshape
	#1]\ignorespaces
}{%
	\endtrivlist\@endpefalse
}
\providecommand{\proofname}{Proof}
\fi
\theoremstyle{plain}

\providecommand{\lemmaname}{Lemma}
\providecommand{\definitionname}{Definition}
\providecommand{\propositionname}{Proposition}

\usepackage{babel}
\usepackage{txfonts}%\usepackage{braket}
\usepackage{colortbl}\definecolor{myurlcolor}{rgb}{0,0,0.7}
\newcommand{\id}{{\operatorname{id}}}

\def\ket#1{| #1 \rangle}
\def\bra#1{\langle  #1 |}

\usepackage[normalem]{ulem}
\usepackage{tikz}
\usetikzlibrary{circuits.ee.IEC}

\newcommand{\mS}{\mathcal{S}}
\newcommand{\sH}{\mathcal{H}}
\newcommand{\mE}{\mathcal{E}}
\newcommand{\mR}{\mathcal{R}}

\newcommand{\haH}

\newtheorem{theorem}{Theorem}

\newtheorem{definition}{Definition}

\newcommand{\sq}{{\operatorname{sq}}}

\definecolor{orange}{RGB}{255,127,0}

\newcommand{\PTr}[2]{\operatorname{Tr}_{#1}\!\left\{{#2}\right\}}

\newcommand{\mF}{\mathcal{F}}
\newcommand{\mD}{\mathcal{D}}

\renewcommand{\leq}{\leqslant}
\renewcommand{\geq}{\geqslant}
\renewcommand{\le}{\leqslant}
\renewcommand{\ge}{\geqslant}

\begin{document}
\title{Causal and Non-Causal Revivals of Information:\\A New Regime of Non-Markovianity in Quantum Stochastic Processes}

\author{Francesco Buscemi}
\email{buscemi@nagoya-u.jp}
\affiliation{Department of Mathematical Informatics, Nagoya University, Furo-cho, Chikusa-ku, 464-8601 Nagoya, Japan}

\author{Rajeev Gangwar}
\affiliation{Department of Physical Sciences, Indian Institute of Science Education and Research (IISER), Mohali, Punjab 140306, India}

\author{Kaumudibikash Goswami}
\email{goswami.kaumudibikash@gmail.com}
\affiliation{QICI Quantum Information and Computation Initiative, Department of Computer Science, The University of Hong Kong, Pokfulam Road, Hong Kong}

\author{Himanshu Badhani}
\affiliation{The Institute of Mathematical Sciences, C. I. T. Campus, Taramani, Chennai 600113, India}
\affiliation{Homi Bhabha National Institute, Training School Complex, Anushakti Nagar, Mumbai 400094, India}
\affiliation{Center for Security, Theory and Algorithmic Research (CSTAR), Centre for Quantum Science and Technology (CQST), International Institute of Information Technology, Hyderabad, Gachibowli, Telangana 500032, India}

\author{Tanmoy Pandit}
\affiliation{Fritz Haber Research Center for Molecular Dynamics, Hebrew University of Jerusalem, Jerusalem 9190401, Israel}

\author{Brij Mohan}
\affiliation{Department of Physical Sciences, Indian Institute of Science Education and Research (IISER), Mohali, Punjab 140306, India}

\author{Siddhartha Das}
\email{das.seed@iiit.ac.in}
\affiliation{Center for Security, Theory and Algorithmic Research (CSTAR), Centre for Quantum Science and Technology (CQST), International Institute of Information Technology, Hyderabad, Gachibowli, Telangana 500032, India}

\author{Manabendra Nath Bera}
\email{mnbera@gmail.com}
\affiliation{Department of Physical Sciences, Indian Institute of Science Education and Research (IISER), Mohali, Punjab 140306, India}

\begin{abstract}
The study of information revivals, witnessing the violation of certain data-processing inequalities, has provided an important paradigm in the study of non-Markovian quantum stochastic processes. Although often used interchangeably, we argue here that the notions of ``revivals'' and ``backflows'', i.e., flows of information from the environment back into the system, are distinct: an information revival can occur without any backflow ever taking place. In this paper, we examine in detail the phenomenon of non-causal revivals and relate them to the theory of short Markov chains and squashed non-Markovianity. We also provide an operational condition, in terms of system-only degrees of freedom, to witness the presence of genuine backflow that cannot be explained by non-causal revivals. As a byproduct, we demonstrate that focusing on processes with genuine backflows, while excluding those with only non-causal revivals, resolves the issue of non-convexity of Markovianity, thus enabling the construction of a convex resource theory of genuine quantum non-Markovianity.

\end{abstract}

\maketitle

\section{Introduction}

The evolution of a system interacting with its surrounding environment, also known as open quantum system dynamics~\cite{Breuer_book_OQS2007}, can be divided into two main categories: Markovian and non-Markovian. Broadly speaking, there are two approaches to the study of quantum non-Markovianity, proposed separately by Rivas--Huelga--Plenio (RHP)~\cite{Rivas2010} and Breuer--Laine--Piilo (BLP)~\cite{Breuer2009}. According to RHP, the evolution of an open quantum system is Markovian if it is \textit{divisible}, i.e., if it can be represented as a composition of steps from one moment in time to the next. This approach is conceptually very close to the idea of semigroups~\cite{alicki-lendi-semigroups}. The other approach, proposed by BLP, is based on the distinguishability of quantum states. Given that Markovian dynamics never increases the distinguishability of quantum states in time, BLP propose that non-Markovian dynamics is associated with an increase in distinguishability, which is interpreted as \textit{revival} of information~\cite{Breuer2009,Buscemi2016a,megier-prl-2021,Smirne2021connectionbetween,vacchini-2022-review,vacchini-2024-comparison-review}.

Thus, by interpreting a decrease in distinguishability as a loss of information, the emergence of non-Markovian effects, according to BLP, seems to imply a \textit{backflow} of information from the environment into the system. This idea later motivated the study of non-Markovianity not only in terms of increase in distinguishability but more generally in terms of revival of correlations between the system, the environment, and a reference, where correlations are quantified using mutual information~\cite{Luo2012}, conditional mutual information~\cite{Huang2021}, entanglement~\cite{Kolodynski2020,das2018fundamental}, interferometric power~\cite{Dhar2015, Souza2015}, and other variants~\cite{Chen2016,das2018fundamental}.

The notion of information revival has also been explored in situations where the environment is limited to some special forms of physical interest~\cite{Alipour2012, Fanchini2014, Haseli2014, Santis2020}. For example, for \textit{classical} environments, it has been shown that a revival of information is possible with the environment unaffected by the system's dynamics~\cite{Franco2012, Budini2018}. Furthermore, attempts have been made to distinguish the classical and quantum contributions to information revival and non-Markovianity~\cite{Milz2020, Giarmatzi2021, Banacki2023, Backer2024}. Nevertheless, the terms ``revivals'' and ``backflows'' are still used interchangeably in the literature.

In this paper, we argue that the concepts of information revival and information backflow can be distinguished and that distinguishing between them can lead to a clearer understanding of what non-Markovianity actually means. More specifically, we find that it is possible for revivals to occur without any actual backflow of information from the environment into the system. Upon closer examination from a causal perspective, we find that such revivals are \textit{non-causal}, in the sense that they can be ``unraveled'' using hidden degrees of freedom that have never interacted with the system and are thus completely causally separated from it. We derive an information-theoretic condition, formulated in terms of quantum conditional mutual information, which is equivalent to an arbitrary system-environment interaction giving rise only to non-causal revivals. This condition is related to the concept of squashed non-Markovianity, a notion recently introduced in~\cite{Gangwar2024}. Furthermore, we show that the condition we identify is robust to small errors and respects convexity.

In contrast to non-causal revivals, a dynamics has a \textit{genuine backflow} if the revival cannot be explained without resorting to some information flowing back and forth from the environment to the system. Remarkably, in some situations our information-theoretic approach allows us to conclude that an information revival involves a genuine backflow of information just by looking at the system's degrees of freedom, without any reference to the environment. In particular, the fine-grained notion of information revivals that we propose here also allows us to solve the well-known problem related to the non-convex character of Markovianity~\cite{Chitambar2019}: while convex mixtures of Markovian processes can lead to non-Markovianity, convex mixtures of non-causal revivals are necessarily non-causal. This observation provides the basis for developing a convex resource theory for non-Markovianity with genuine backflows, after non-causal revivals have been removed.

The article is structured as follows: In Sec.~\ref{sec:Background}, we present a generic model of dynamical processes involving a system and its environment, where the reduced dynamics of the system can be either Markovian or non-Markovian, depending on the presence or absence of information revival. Information revival, often linked to the violation of the data-processing inequality, is explained in Sec.~\ref{sec:InfoRevival} with a motivating example. Sec.~\ref{sec:non-CausalRevivals} addresses the distinction between information revival and information backflow from the environment to the system in non-Markovian evolutions and characterizes non-causal revivals without information backflow. The conditions for non-causal revivals are generalized in Sec.~\ref{sec:GeneralCase}. In Sec.~\ref{sec:OpCondition}, we introduce operational conditions for genuine information backflow, which cannot be attributed to non-causal revivals, in terms of system-only degrees of freedom. Sec.~\ref{sec:ConvexResource} shows that quantum non-Markovianity based on genuine information backflow is convex, suggesting the potential for developing a consistent resource theory framework to characterize quantum non-Markovian dynamics as quantum resource. Finally, we conclude in Sec.~\ref{sec:Conclusion}.

\section{Background: processes VS models \label{sec:Background}}

The setting is the usual one in open systems dynamics~\cite{Breuer_book_OQS2007}: a quantum system $Q$ with $d$-dimensional ($d<\infty$) Hilbert space $\sH_Q$, which can be initialized in any density operator (i.e., state) $\rho_Q\in\mS(\sH_Q)$, interacts with an environment (also assumed to be finite-dimensional), which is instead initialized in an arbitrary but fixed state $\gamma_E\in\mS(\sH_E)$. The interaction is modeled as a bipartite unitary operator $U_{QE}:\sH_Q\otimes\sH_E\to \sH_{Q'}\otimes\sH_{E'}$, which can be parameterized by time, so to describe a joint evolution. In general, since the system is open and particles can be exchanged, we only assume that $\sH_Q\otimes\sH_E\cong \sH_{Q'}\otimes\sH_{E'}$, while the local dimensions are allowed to change. As we impose no restrictions on the system's initial state, following a common convention in quantum information theory, we introduce a reference system $R$ with $\sH_R\cong\sH_Q$, and assume that their initial state is pure and maximally entangled, i.e., $\Phi^+_{RQ}$ such that $\PTr{R}{\Phi^+_{RQ}}=d^{-1}\openone_Q$~\cite{note_reference}.

After introducing the reference system, we consider three ``snapshots'' of the joint tripartite reference-system-environment configuration, taken at three different times $t_0<t_1<t_2$. Without loss of generality, we can write 
\begin{align}
    &\rho_{RQE}=\Phi^+_{RQ}\otimes\gamma_E &t=t_0\;,\label{eq:3-point-conf-0}\\ 
    \xrightarrow{t_1}\ &\sigma_{RQ'E'}=U_{QE}\;\rho_{RQE}\;U^\dag_{QE} &t=t_1\;,\label{eq:3-point-conf-1}\\
    \xrightarrow{t_2}\ &\tau_{RQ''E''}=V_{Q'E'}\;\sigma_{RQ'E'}\;V^\dag_{Q'E'} &t=t_2\;.\label{eq:3-point-conf-2}
\end{align}
The unitaries $U_{QE}$ and $V_{Q'E'}$ govern the time evolutions in the first ($t_0 \to t_1$) and the second ($t_1 \to t_2$) steps, respectively. In what follows, we will call the above sequence of tripartite states a \textit{three-time snapshot} for the dynamics at hand: it provides the minimal framework for discussing information revivals in both discrete and continuous time. Generalizations to more than three points in time are straightforward, but for the purposes of the present discussion, three-point snapshots will suffice.

By tracing over the environment, we obtain the reduced reference-system dynamics
\begin{align}\label{eq:open-three-time}
    \Phi^+_{RQ}\xrightarrow{t_1} \sigma_{RQ'}\xrightarrow{t_2} \tau_{RQ''}\;.
\end{align}
Since the initial state of the environment is fixed, it is possible to represent the above sequence using the formalism of \textit{quantum channels}, i.e., completely positive trace-preserving (CPTP) linear maps: the theory guarantees the existence of two quantum channels $\mE_{Q\to Q'}$ and $\mF_{Q\to Q''}$ such that $\sigma_{RQ'}=(\id_R\otimes\mE_{Q\to Q'})(\Phi^+_{RQ})$  and $\tau_{RQ''}=(\id_R\otimes\mF_{Q\to Q''})(\Phi^+_{RQ})$. However, the existence of an intermediate channel $\mD_{Q'\to Q''}$ such that $\tau_{RQ''}=(\id_R\otimes\mD_{Q'\to Q''})(\sigma_{RQ'})$ is \textit{not} guaranteed, since at time $t=t_1$ system and environment are in general correlated~\cite{Pechukas1994,Buscemi2014}. But if such a channel exists, then the three-time open system dynamics in~\eqref{eq:open-three-time} is called \textit{divisible}.

From an operational point of view, the reduced dynamics in Eq.~\eqref{eq:open-three-time} is all that is directly accessible to the observer. In other words, while Eq.~\eqref{eq:open-three-time} provides the \textit{operational} description of the process, the three-time snapshot in Eqs.~\eqref{eq:3-point-conf-0}--\eqref{eq:3-point-conf-2} provides an \textit{epistemic} model of it. Correspondingly, we say that a three-time model~\eqref{eq:3-point-conf-0}--\eqref{eq:3-point-conf-2} is \textit{consistent} with a three-time process~\eqref{eq:open-three-time} if the latter is obtained by taking the partial trace of the former. Of course, many different epistemic models can be consistent with the same operational process. This distinction between the operational process (i.e., the system's reduced dynamics) and the epistemic model (i.e., the specific choice of system--environment interaction mechanism) plays an important role in what follows.

We focus on three entropic measures of information: von Neumann entropy, quantum mutual information (QMI), and quantum conditional mutual information (QCMI), which are defined as follows. Given a system $A$ in state $\rho_A$, its von Neumann entropy is $H(A)_{\rho}=-\Tr{\rho_A\log\rho_A}$. Given a bipartite system $AB$ in state $\rho_{AB}$, its QMI is $I(A;B)_{\rho}=H(A)_{\rho}+H(B)_{\rho}-H(AB)_{\rho}$. The QMI $I(A;B)$ provides an operationally well-defined measure of the total amount of correlations existing between systems $A$ and $B$~\cite{Groisman2005}. Finally, given a tripartite system $ABC$ in state $\rho_{ABC}$, its QCMI is $I(A;C|B)_{\rho}=I(A;BC)_{\rho}-I(A;B)_{\rho}=H(AB)_\rho+H(BC)_{\rho}-H(B)_{\rho}-H(ABC)_{\rho}$. The von Neumann entropy, QMI, and QCMI are always nonnegative for quantum states. Optimizations of QCMI lead to the definitions of squashed entanglement~\cite{christandl-winter-2004-squashed}, squashed non-Markovianity~\cite{Gangwar2024}, and puffed entanglement~\cite{oppenheim2008paradigm}: see Appendix~\ref{App:squashed-puffed} for definitions and properties of these quantities, which will be used in what follows.

\section{Information revivals \label{sec:InfoRevival}}
Crucially, QMI satisfies the \textit{data-processing inequality}: for any bipartite state $\rho_{AB}$ and any channel $\mE_{B\to B'}$, the QMI $I(A;B')$ computed for $(\id_A\otimes\mE_{B\to B'})(\rho_{AB})$ obeys the inequality $I(A;B')\le I(A;B)$. The idea is that the correlations between two systems cannot increase as a consequence of (deterministic) local actions: a very natural requirement to be satisfied by any reasonable measure of correlations, lest the notion of locality itself be violated. 

Therefore, in any three-time sequence such as~\eqref{eq:open-three-time}, it will always hold that $I(R;Q')_1\le I(R;Q)_0$ and $I(R;Q'')_2\le I(R;Q)_0$. However, since, as we already noticed, at time $t=t_1$ system and environment are generally correlated, a channel from $Q'$ to $Q''$ may not exist, and we may observe a \textit{revival} of QMI, i.e.,
\begin{align}\label{eq:vanilla-revival}
     I(R;Q'')_2> I(R;Q')_1\;.
\end{align}
Whenever a revival occurs, we can take this as a conclusive signature for the fact that the evolution between $t_1$ and $t_2$ is non-Markovian~\cite{Luo2012}.

Before continuing our analysis, we briefly comment on the relationship between witnesses of non-Markovianity based on correlation measures~\cite{Luo2012,Buscemi2014}, like that in Eq. ~\eqref{eq:vanilla-revival}, and those based on distinguishability measures, as proposed in a seminal proposal by Breuer, Laine, and Piilo~\cite{Breuer2009}, and later developed in several papers~\cite{vacchini-2022-review,vacchini-2024-comparison-review}. Indeed, distinguishability measures are also correlation measures, albeit only classical ones: as shown in~\cite{Konig}, the logarithm of the guessing probability is the (negative of the) \textit{min-conditional entropy}---a close relative of QMI---computed with respect to a classical-quantum state representing the classical correlation between a register and the ``labels'' attached to each state to be distinguished. The equivalence between distinguishability-based witnesses, correlation-based witnesses, and divisibility has been proved in~\cite{Buscemi2016a,buscemi2018reverse-data-proc}.

Note that the presence or absence of a revival depends solely on the reduced dynamics of the system, and is thus fully operational. Nevertheless, depending on the nature of the revival, different system--environment interaction mechanisms may or may not be possible. In this work, our main objective is to investigate how the QMI between reference and system changes between times $t_1$ and $t_2$, and to relate such changes to the algebraic and information-theoretic properties of the system--environment interaction models compatible with the process, in particular, the intermediate tripartite configuration $\sigma_{RQ'E'}$ in Eq.~\eqref{eq:3-point-conf-1}.

\subsection{Explaining revivals}

Whenever a revival of correlations as in~\eqref{eq:vanilla-revival} occurs as a consequence of a \textit{local} operation on the system, instead of immediately concluding that some fundamental law of nature has been violated, it is more natural to explain the observed revival by assuming the existence of other degrees of freedom which, although interacting with the system, were not included in the balance, thus leading to an apparent revival. Accordingly, an \textit{explanation} for a revival consists of incorporating into the balance other degrees of freedom, compatible with the overall three-time process~\eqref{eq:open-three-time}, until the revival disappears.

An obvious way to explain \textit{any} revival is to consider a system-environment interaction mechanism such as \eqref{eq:3-point-conf-0}--\eqref{eq:3-point-conf-2} and include the environment itself in the balance. More precisely, instead of comparing only the correlation content of $Q'$ versus that of $Q''$, we compare $Q'E'$ versus $Q''E''$ as a whole. When we do this, since the joint system-environment dynamics is unitary, we have $I(R;Q'E')_1=I(R;Q''E'')_2$, and the anomalous revival is naturalized. Again, for the sake of conceptual clarity, we emphasize that, while the revival is operational, the explanation given for it is epistemic and depends on the system--environment interaction model chosen.

This way of explaining revivals (i.e., by incorporating the environment) lies behind the interpretation of revivals as \textit{backflows} of information from the environment into the system. The idea is that some of the initial correlations between the system and the reference were moved to the environment as a consequence of the interaction between $t_0$ and $t_1$, and later restored at time $t_2$. Therefore, the violation of data-processing inequality can be explained by tracking such displaced correlations as they move back and forth between the system and the environment.

\subsection{A motivating example}\label{sec:example}

While all revivals can be explained as backflows, not all revivals \textit{require} backflows to be explained. To illustrate this point, let us consider a concrete example, similar to those discussed in Ref.~\cite{Franco2012}, in which $\sH_Q\cong\sH_R\cong\mathbb{C}^2$, $\sH_E\cong\mathbb{C}^4$, and $\gamma_E=\openone_E/4$. The interaction between the system and environment is modeled as the repeated application of the same control-unitary operator $U_{QE}=\sum_{i=0}^3\pi^i_Q\otimes\ketbra{i}_E$, where $\pi^1_Q=X_Q$, $\pi^2_Q=Y_Q$, and $\pi^3_Q=Z_Q$ are the Pauli matrices and $\pi^0_Q=\openone_Q$.

The resulting three-time process, obtained by tracing over the environment, can be easily computed as follows:
$ \Phi^+_{RQ}\xrightarrow{t_1} \frac{\openone_R}{2}\otimes \frac{\openone_{Q'}}{2}\xrightarrow{t_2}\Phi^+_{RQ''}\;.$ Correspondingly, the correlation between the reference and the system is maximal at initial time $t_0$, vanishes at intermediate time $t_1$, and is maximal again at final time $t_2$, i.e., $I(R;Q)_{0}=2\xrightarrow{t_1} I(R;Q')_{1}=0\xrightarrow{t_2} I(R;Q'')_{2}=2$ bits. This model thus exhibits a complete revival of information. Such a revival, however, can be explained \textit{without} the need for a backflow. Such an explanation can be given, for example, by including in the picture an ancillary system $\widetilde{E}$, which is perfectly correlated with the environment but never interacts with the system and thus cannot give back to it any information. Nonetheless, such an ancillary system $\widetilde{E}$ is such that $I(R;Q\widetilde{E})_0=I(R;Q'\widetilde{E})_1=I(R;Q''\widetilde{E})_2=2$ bits, i.e., the revival is explained. A detailed discussion can be found in Appendix~\ref{app:revival}.

As a consequence, it is clear that in this case the violation of the data-processing inequality, i.e., $I(R;Q'')_2>I(R;Q')_1$, is nothing but an artifact due to the ignorance of the information residing in $\widetilde{E}$: information that was already there before $Q$ and $E$ interacted, and is completely independent of both $Q$ and $R$. Hence, we conclude that in this case no backflow of information can be inferred, despite the fact that a revival is observed. In what follows, we will make this idea more rigorous and characterize other situations, beyond the highly idealized example above, in which revivals can explained without resorting to a backflow.

\section{Non-causal revivals \label{sec:non-CausalRevivals}}

In order to generalize the above example, let us consider a three-time snapshot~\eqref{eq:3-point-conf-0}--\eqref{eq:3-point-conf-2} and extend its initial configuration as $\rho_{RQEF}=\Phi^+_{RQ}\otimes\gamma_{EF}$, where $\Tr_F[\gamma_{EF}]=\gamma_{E}$. Notice that $\gamma_{EF}$ may be mixed. The extension $F$ does not participate in the process, i.e., it is only acted upon by the identity operator, similar to the reference $R$, and like the latter it is more of a mathematical device than an actual physical system. In this case, we say that the extension $F$ is \textit{inert} from $t_0$ to $t_2$. As a consequence, the (unitarily) evolved states after the first and second steps, i.e., $\sigma_{RQ'E'F}$ and $\tau_{RQ''E''F}$, respectively, are automatically extensions of $\sigma_{RQ'E'}$ and $\tau_{RQ''E''}$, respectively; in formula, $\Tr_F[\sigma_{RQ'E'F}]=\sigma_{RQ'E'}$ and $\Tr_F[\tau_{RQ''E''F}]=\tau_{RQ''E''}$.

As mentioned earlier, without loss of generality, we shall focus on the change of quantum mutual information occurring in the second step ($t_1 \to t_2$).

\begin{definition}[Non-causal revivals]\label{def:non-causal-rev}
We say that an open system dynamics~\eqref{eq:open-three-time} exhibits a \emph{non-causal revival} of information, whenever a revival occurs, i.e., $I(R;Q'')_2>  I(R;Q')_1$, but there exists a compatible model, such as that in Eqs.~\eqref{eq:3-point-conf-0}--\eqref{eq:3-point-conf-2}, and an inert extension $F$ such that
\begin{align}\label{eq:def-sq-rev}
I(R;Q''F)_2\le I(R;Q'F)_1\;.
\end{align}
\end{definition}

The use of the term ``non-causal'' in Definition~\ref{def:non-causal-rev} is justified because revivals that can be explained in terms of an inert extension do not, by construction, require any direct backflow of information in order to be explained. Equivalently, non-causal revivals are precisely those that can be naturalized by resorting to a system that is space-like separated, and thus causally separated, from both the system and the environment.

While the property of being non-causal does not strictly belong to the dynamics of the reduced system, but to the epistemic model given for it, we nevertheless call the former non-causal whenever it can be explained by a non-causal model. This is entirely analogous to the discussion of, e.g., Bell's inequalities, where the observed correlations are called ``classical'' whenever there is a classical mechanism that could have given rise to them, even though they actually arise from genuinely quantum, albeit noisy, devices.

\subsection{Causal and non-causal models}

Suppose that an interaction model~\eqref{eq:3-point-conf-0}--\eqref{eq:3-point-conf-2} is given, in which the environment starts in a pure state, i.e., the density operator $\gamma_E$ in~\eqref{eq:3-point-conf-0} is rank-one. Then any inert extension must be trivial, i.e., in tensor product with all the rest. In this case, $I(R;Q''F)_2\le I(R;Q'F)_1$ if and only if $I(R;Q'')_2\le I(R;Q')_1$, i.e., if and only if there was no revival to begin with. In other words, any revival that occurs in the presence of a pure environment, cannot be non-causal, but requires a genuine backflow from the environment in order to be explained. Vice versa, if a revival admits a non-causal explanation, then such an explanation will necessarily involve an interaction model with a mixed-state environment. In fact, in most cases of theoretical and experimental interest, the environment is usually assumed to be at some finite temperature, thus leaving open the possibility of non-causal revivals.

The above discussion raises the question: if for an initially pure environment all revivals can only be exaplained as backflows, are there situations in which, on the contrary, \textit{all} revivals are non-causal?

In order to answer this question, let us notice that,
since the joint system-environment evolution is unitary, the quantum mutual information between the reference and the rest never changes, i.e., $I(R;QEF)_0 =I(R;Q'E'F)_1=I(R;Q''E''F)_2$. Therefore, the inequality~\eqref{eq:def-sq-rev} can also be cast in terms of quantum conditional mutual information (QCMI), so that a revival is non-causal if and only if there exists an inert extension $F$ such that
\begin{align}\label{eq:CMI-diff}
I(R;E''|Q''F)_2\ge I(R;E'|Q'F)_1\;.    
\end{align}
Note that the inequality is reversed as a result of the QMI being replaced by the QCMI. We now recall the fact that the QCMI is always non-negative. Therefore, given an arbitrary three-time snapshot~\eqref{eq:3-point-conf-0}--\eqref{eq:3-point-conf-2}, if there exists an inert extension $F$ such that $I(R;E'|Q'F)_1=0$, then Eq.~\eqref{eq:CMI-diff} is always satisfied, and any revival is non-causal.

Thanks to Uhlmann's theorem for purifications and Stinespring's dilation theorem for CPTP maps~\cite{Stinespring1955, Kraus1983}, we can say that given a mixed state $\rho_A$, any extension $\rho_{AX}$ of it can be obtained by starting from some purification $\ket{\varphi}_{AB}$ and then acting on $B$ with some channel $\mathcal{N}_{B\to X}$.
Furthermore, it can be observed that since the joint reference-system state is initially pure and the overall evolution is unitary in Eqs.~\eqref{eq:3-point-conf-0}--\eqref{eq:3-point-conf-2}, the only mixed component arises from the initial state of the environment, $\gamma_E$.
These two facts together show that any extension done at any point in time, if kept invariant through the process, provides an inert extension of the three-point snapshot.
This implies that, without loss of generality, it is possible to construct an extension for the intermediate configuration $\sigma_{RQ'E'F}$ on its own, regardless of the initial and the final configurations, and such an extension will automatically provide an inert extension for the entire three-point snapshot.

We thus reach the following conclusion, that we state as a theorem:

\begin{theorem}\label{th:zero-squashed}
Suppose that, starting from an initial configuration as in Eq.~\eqref{eq:3-point-conf-0}, the joint reference-system-environment reaches the configuration $\sigma_{RQ'E'}$ at time $t=t_1$. Then, regardless of the next interaction step $V_{Q'E'}$ in~\eqref{eq:3-point-conf-2}, all revivals possibly occurring are non-causal, if and only if there exists an inert extension $F$ such that
\begin{align}\label{eq:sq-non-markov}
    I(R;E'|Q'F)_1=0\;.
\end{align}
\end{theorem}

\begin{proof}
One direction comes from the non-negativity of the QCMI, so that if~\eqref{eq:sq-non-markov} holds, then~\eqref{eq:CMI-diff} also holds. Vice versa, taking $V_{Q'E'}$ to be the interaction merging the entire system $E'$ with $Q'$ into $Q''$, i.e., $\sH_{Q''}\cong\sH_{Q'}\otimes\sH_{E'}$ and $\sH_{E''}\cong\mathbb{C}$, we have $I(R;E''|Q''F)=0$, and again from the positivity of the QCMI, condition~\eqref{eq:CMI-diff} implies~\eqref{eq:sq-non-markov}.
\end{proof}

Note that condition~\eqref{eq:sq-non-markov} implies that the intermediate configuration $\sigma_{RQ'E'}$ has zero squashed non-Markovianity between $R,E'$ conditional on $Q'$. This, in turn, implies that the squashed entanglement of $\sigma_{RE'}$ is zero, i.e., that $\sigma_{RE'}$ is \textit{separable}. See Appendix~\ref{App:squashed-puffed} for definitions and basic facts.

Theorem~\ref{th:zero-squashed} provides another motivation for using the term ``non-causal explanation''. As a consequence of Petz's theory of statistical sufficiency~\cite{Petz1986, Petz1988, Hayden2004}, condition~\eqref{eq:sq-non-markov} is equivalent to the existence of a channel $\mR_{Q'F\to Q'E'F}$ which can reconstruct the state $\sigma_{RQ'E'F}$ from its marginal $\sigma_{RQ'F}=\Tr_{E'}[\sigma_{RQ'E'F}]$. Subsequently, there exists a channel ${\cal N}_{Q'F \to Q''E''F}$, where
\begin{align}
   { \cal N}_{Q'F \to Q''E''F}\coloneqq  {\cal V}_{Q'E'\to Q''E''}\circ {\cal R}_{Q'F \to Q'E'F}, \label{Eq:composition_N}
\end{align}
with ${\cal V}(\cdot)\coloneqq  V(\cdot)V^\dagger$, such that  
\begin{align}
(\id_R\otimes{\cal N}_{{Q'F}\to {Q''E''F}})(\sigma_{RQ'F})=\tau_{RQ''E''F}.    
\end{align}

\noindent 
In other words, the transformation $\sigma_{RQ'}\to\tau_{RQ''}$ that leads to the observed revival can be exactly reproduced even without the environment, using the inert extension $F$, which, as mentioned above, never interacted with the system in the past and thus cannot give anything back to it: the revival is reproduced without any backflow ever occurring. See Fig.~\ref{fig:noncausal}.

\begin{figure}
    \centering
    \includegraphics[width=\columnwidth]{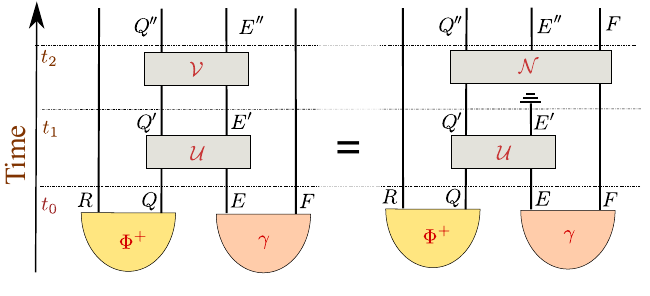}
    \caption{For non-causal models, all revivals are non-causal. At time $t=t_0$ we have the state $\rho_{RQEF}=\Phi^+_{RQ}{\otimes}\gamma_{EF}$. Here we assume that the system ($Q$) and the reference ($R$) are maximally entangled. Moreover, an inert system $F$, which does not participate in the subsequent dynamics, can be correlated with the environment ($E$) in a joint state $\gamma_{EF}$. Under the action of the unitary channel ${\cal U}_{QE}(\cdot)\coloneqq U_{QE}(\cdot)U_{QE}^\dagger$, the evolved state at time $t=t_1$ is $\sigma_{RQ'E'F}={\cal U}_{QE} \, (\rho_{RQEF})$. If $\sigma_{RQ'E'F}$ satisfies $I(R;E'|Q'F)_\sigma=0$, then for each subsequent unitary channel ${\cal V}_{Q'E'}\coloneqq V_{Q'E'}(\cdot)V_{Q'E'}^\dagger$ there exists a corresponding quantum channel ${\cal N}_{Q'F\to Q''E'' F}$, which can exactly reproduce the effect of ${\cal V}$, but without needing $E'$. Thus, an apparent backflow of information from the environment ($E'$) to the system ($Q''$) can be explained from the perspective of $Q'F$ alone, independent of the environment $E'$. The channel ${\cal N}$ is composed of the recovery map ${\cal R}_{Q'F\to Q' E' F}$, which reconstructs the state $\sigma_{RQ'E'F}$ from its marginal $\sigma_{RQ'F}$, followed by the unitary ${\cal {V}}_{Q'E'}$, see Eq.~\eqref{Eq:composition_N}. In this scenario, any information revival that occurs between times $t=t_1$ and $t=t_2$ is necessarily non-causal, as shown in Theorem~\ref{th:zero-squashed}.}
    \label{fig:noncausal}
\end{figure}

\section{The general case \label{sec:GeneralCase}}

Besides the two extreme situations, i.e., one in which the environment is initially pure and all revivals require a corresponding backflow, and the other in which the intermediate configuration satisfies~\eqref{eq:sq-non-markov} and all revivals are non-causal, what can be said about the general case?

In general, this is a very difficult question to answer since the search for the inert extension satisfying Eq.~\eqref{eq:sq-non-markov} must go through systems $F$ with no prior upper bound on their dimensions. Besides the trivial upper bounds
$\inf_FI(R;E'|Q'F)_1\le\min\{I(R;E'|Q')_1, I(R;E')_1\}\;,$
respectively obtained when $F\cong\mathbb{C}$ and when $F$ is a purification of $\sigma_{RQ'E'}$, the general case is very hard: at least as hard as the separability problem, since, as we already noticed, finding an extension that satisfies~\eqref{eq:sq-non-markov} would also show that the state $\sigma_{RE'}$ is separable.
It is therefore crucial, both practically and conceptually, to be able to address the case where one knows only that there exists an inert extension such that $I(R:E'|Q'F)_1\le\varepsilon$, for some given threshold value $\varepsilon\ge 0$. We address this in the following theorem.

\begin{theorem}\label{th:approx}
Given a three-time snapshot, Eqs.~\eqref{eq:3-point-conf-0}--\eqref{eq:3-point-conf-2}, suppose that the intermediate configuration $\sigma_{RQ'E'}$ is such that there exists an inert extension $F$ with
\begin{align}\label{eq:approx}
I(R;E'|Q'F)_1\leq \varepsilon\;,   
\end{align}
for small value $\varepsilon\ge 0$. Then, for any subsequent interaction $V_{Q'E'}$, there exists a corresponding channel $\mathcal{N}_{Q'F\to Q''E''F}$ able to provide an approximate non-causal explanation of the reduced dynamics from $t_1$ to $t_2$, in formula,
$
\mathsf{F}\Big(\tau_{R Q''E''F}, \mathcal{N}_{Q'F\rightarrow Q''E''F} \left(\sigma_{RQ'F}\right)\Big) \geq 2^{-\varepsilon}$ \;,
where $\mathsf{F}(\alpha,\beta)\coloneqq \left(\operatorname{Tr}\sqrt{\sqrt{\alpha}\beta\sqrt{\alpha}}\right)^2$ is the (squared) fidelity between states $\alpha$ and $\beta$.  
\end{theorem}

The proof of the above theorem is the consequence of Theorem 7 in \cite{Buscemi2016}, and the converse is also known to hold, as in Theorem 8 in \cite{Buscemi2016}. Thus, we conclude that the squashed non-Markovianity of the intermediate configuration $\sigma_{RQ'E'}$ is a good indicator of how non-causal any subsequent revival can be, in the sense that small squashed non-Markovianity guarantees that any revival can be ``almost'' explained by an inert extension. Notice that, in particular, condition~\eqref{eq:approx} implies that $I(R;Q''F)_2-I(R;Q'F)_1=I(R;E'|Q'F)_1-I(R;E''|Q''F)_2\le\varepsilon$, regardless of the magnitude of the actual revival $I(R;Q'')_2-I(R;Q')_1$.

\section{Operational condition for the presence of genuine backflow \label{sec:OpCondition}}
As emphasized earlier, while an information revival is operational, its interpretation (as non-causal revival or genuine backflow) depends on the epistemic model given for the interaction. Nevertheless, there may be situations in which the presence of genuine backflow can be inferred in a model-independent way, similar to Bell's inequalities, which, if violated, guarantee the presence of quantum correlations, regardless of how the experiment was actually conducted.

Thus, we now consider the cases where the revival of information cannot entirely be explained by a non-causal revival and thus necessarily contains at least ``some'' genuine backflow of information. Let us recall Definition~\ref{def:non-causal-rev}: if there does not exist an inert extension $F$ that satisfies Eq.~\eqref{eq:def-sq-rev}, then we have a scenario where the revivals are not non-causal revivals. Equivalently:
\begin{align}
    I(R;Q''F)_2 > I(R;Q'F)_1, \quad \forall \ \mathrm{ inert} \ F\;.\label{eq:not_noncausal}
\end{align}
Using the chain rule of QCMI and considering that the quantum mutual information between $R$ and $F$ does not change, the above condition~\eqref{eq:not_noncausal} can be reformulated in terms of QCMI as:
\begin{align}
    I(R;Q''|F)_2 > I(R;Q'|F)_1 , \quad \forall \ \mathrm{ inert} \ F\;. \label{eq:not_noncausal_QCMI}
\end{align}
The above is, by definition, a necessary and sufficient condition for a revival to contain genuine information backflow. However since the notion of ``inert extension'' depends on the epistemic model chosen for the interaction, the condition indirectly depends on the knowledge of the environment at $t_0$. In what follows, we remove such a dependence from condition~\eqref{eq:not_noncausal_QCMI} and use it to derive an \textit{operational} (i.e., independent of the model chosen) condition, sufficient (albeit, in general, not necessary) to guarantee the presence of a non-zero amount of genuine backflow of information, without any reference to the environment. In what follows, we will make use of the concepts of \textit{squashed} and \textit{puffed entanglement}; see Appendix~\ref{App:squashed-puffed} for details.

\begin{theorem}\label{th:not_non-causal}
Given a three-time process as in~\eqref{eq:open-three-time}, if
\[
E_{\mathrm{sq}}(R;Q'')_2 > H(Q') _1\;,
\]
then the revival requires a non-zero amount of genuine backflow to be explained (i.e., it is causal).
\end{theorem}

\begin{proof}
We start by noticing that a \textit{sufficient} condition for Eq.~\eqref{eq:not_noncausal_QCMI} to hold is that
    \begin{align}
        E_{\mathrm{sq}}(R;Q'')_2 &\coloneqq \dfrac{1}{2} \inf_{G}I(R;Q''|G)_2 \nonumber \\ & > \dfrac{1}{2} \sup_{G}I(R;Q'|G)_1\nonumber\\ &\eqqcolon E_{\mathrm{puff}}(R;Q')_1\;. \label{eq:not_non-causal_iners_Ent}
    \end{align}
Combining the above with Eq.~\eqref{eq:puffed-bound} in Appendix~\ref{App:squashed-puffed}, we see that a sufficient condition for Eq.~\eqref{eq:not_noncausal_QCMI} is also given by
    \begin{align}
        E_{\mathrm{sq}}(R;Q'')_2 > H(Q')_1\;. \label{eq:Theorem_inquality}
    \end{align} 
Thus, whenever the inequality in Eq.~\eqref{eq:Theorem_inquality} is satisfied, genuine backflow, as in Eq.~\eqref{eq:not_noncausal_QCMI}, is guaranteed. This completes the proof.
\end{proof}

A very similar condition was obtained in Ref.~\cite{Backer2024}. However, while the condition proved in Ref.~\cite{Backer2024} is sufficient to falsify the hypothesis of a purely \textit{classical} environment, our condition excludes a purely \textit{non-causal} information revival, which is, in general, a stricter requirement. Nevertheless, this similarity suggests that there may be a close relationship between the notions of classical environments on the one hand and non-causal revivals on the other. We leave this issue open for future investigation.

Motivated by Theorem~\ref{th:not_non-causal}, now we provide an example where the presence of genuine backflow is unavoidable, regardless of the system--environment model chosen. Let us recall the three-time snapshot scenario Eqs.~\eqref{eq:3-point-conf-0}-\eqref{eq:3-point-conf-2}, and consider the initial state of the composite reference-system-environment configuration given by 
\begin{align}
    \Psi_{RQE}=\Phi^+_{RQ} \otimes \gamma_E, \quad t= t_0
\end{align}
where $R$ is the reference system, $Q$ is the system of interest, $\Phi^+_{RQ}$ is a maximally entangled state, and $E$ is the environment. Let us assume that the dimension of all the three systems involved is $d$, while $H(E)_0$ is strictly less than $\log d$.

Let us now suppose that the interaction $t_0\to t_1$ causes a perfect swap between $Q$ and $E$, and that such swap is undone between $t_1$ and $t_2$. Then, the resulting process (obtained after tracing over the environment) is:

\begin{align}\label{eq:genuine-process}
    \Phi_{RQ}^{+} \to \rho_R \otimes \gamma_{Q'} \to \Phi^{+}_{RQ''}\;. 
\end{align}
Since $\Phi_{RQ}^{+}$ is the maximally entangled state, $E_{\rm{sq}}(R;Q)_0=\log d$. The sufficient condition in Theorem~\ref{th:not_non-causal} is thus easily verified because $E_{\rm{sq}}(R;Q'')_2=E_{\rm{sq}}(R;Q)_0=\log d > H(E)_0=H(Q')_1$ is satisfied in this dynamics. Thus, a process like that in Eq.~\eqref{eq:genuine-process} \textit{must} involve some genuine exchange of information between system and environment.

\section{Towards a convex resource theory of non-Markovianity with genuine backflow \label{sec:ConvexResource}}

Let us now consider two processes, each with its own three-time snapshot, i.e.,
\[
\Phi^+_{RQ}\otimes\gamma_E^{(a)}\xrightarrow{t_1} \sigma_{RQ'E'}^{(a)}\xrightarrow{t_2} \tau^{(a)}_{RQ''E''}\;,
\]
and
\[
\Phi^+_{RQ}\otimes\gamma_E^{(b)}\xrightarrow{t_1} \sigma_{RQ'E'}^{(b)}\xrightarrow{t_2} \tau^{(b)}_{RQ''E''}\;,
\]
and let us assume both to be without revival, i.e., $I(R;Q')_1^{(a)}\ge I(R;Q'')_2^{(a)}$ and $I(R;Q')_1^{(b)}\ge I(R;Q'')_2^{(b)}$. And yet, if we consider the process obtained by convexly mixing the two, i.e.,
$
\Phi^+_{RQ}\otimes\sum_xp_x\gamma_E^{(x)}\xrightarrow{t_1} \sum_xp_x\sigma_{RQ'E'}^{(x)}\xrightarrow{t_2} \sum_xp_x\tau^{(x)}_{RQ''E''}\;,    
$
a revival may occur. This is a well-known problem with any attempt to formulate a resource theory of non-Markovianity: Markovian processes do not form a convex set~\cite{Breuer_2018-mixing-induced}.

However, if we focus on three-point snapshots with non-causal revival instead of just snapshots without revival~\cite{Note1}, then convexity is satisfied. This fact is easily shown by choosing, as the inert extension, a classical system perfectly correlated with the index $x\in\{a,b\}$. The detailed argument is shown in Appendix~\ref{App:convexity}. In other words, when mixing Markovian processes, even if a revival may occur between $t_1$ and $t_2$ as a consequence of mixing~\cite{Breuer_2018-mixing-induced}, it will necessarily be non-causal. For the same reason, mixtures of processes with non-causal revivals will also be non-causal, simply because the explanatory inert extension can always include the classical randomness used for mixing.

This observation opens up the possibility of constructing a convex resource theory of dynamical non-Markovianity, defined in terms of \textit{genuine} backflows of information from the environment to the system, while excluding non-causal information revivals. Another possibility is to consider the situation where the system and the reference do not start in the maximally entangled state, but in an arbitrary bipartite mixed state, as done in~\cite{Buscemi2014,megier-entropy-2022}. Furthermore, our findings appear to be related to the problem of how conditional mutual information propagates beyond the causal light cone in an open quantum system~\cite{lee2024universal}. These and other lines of research will be explored elsewhere.

\section{Conclusion \label{sec:Conclusion}}
Non-Markovianity in an open quantum dynamics of a system is often characterized by violations of data-processing inequalities, which are interpreted as information revival and commonly associated with the backflow of information from the environment to the system. In this article, we argued that not all instances of information revival are due to information backflow. To clarify this distinction, we introduced the concept of non-causal revivals, which refer to revivals in non-Markovian quantum processes that can be explained by the presence of an auxiliary system, remaining causally separate from the process. These revivals are artifacts that arise from neglecting the auxiliary system, which may be part of the overall system but does not participate in its evolution.

We presented a necessary and sufficient condition for identifying non-causal revivals using squashed quantum non-Markovianity, which measures genuine quantum non-Markovianity or conditional bipartite entanglement in tripartite quantum systems. This condition is robust under small deviations, ensuring the reliability of the notion of non-causal information revivals. In contrast, any information revival that cannot be explained by non-causal revivals is attributed to genuine information backflow from the environment to the system, indicating the presence of genuine non-Markovianity in dynamics. We derived an operational condition for witnessing genuine backflow and compared it with existing results in the literature. Furthermore, we demonstrated that a probabilistic mixture of Markovian processes can only result in non-causal information revivals and cannot exhibit genuine non-Markovianity. Consequently, the genuine non-Markovianity is a convex property, resolving a long-standing open problem in the field. This finding opens up the possibility of developing a convex resource theory to characterize genuine non-Markovianity as a quantum resource.

In conclusion, our study improves the fundamental understanding of quantum non-Markovianity in open quantum dynamics. Particularly, it demonstrates that there can be apparent non-Markovianity without any backflow of information from the environment to the system and prescribes operational conditions for genuine non-Markovianity. In contrast to common understanding, we have shown that convex mixtures of Markovian dynamics can exhibit information revivals, but this form of non-Markovianity is not genuine. Opening up the possibility of a convex resource theory of genuine non-Markovianity, our results enable a systematic characterization of non-Markovianity as a resource and may find important implications in several closely related areas, including quantum thermodynamics, quantum communication, quantum dynamics, and causality.

\medskip
\textit{Acknowledgments.}---F.B. acknowledges support from MEXT Quantum Leap Flagship Program (MEXT QLEAP) Grant No.~JPMXS0120319794, from MEXT-JSPS Grant-in-Aid for Transformative Research Areas (A) ``Extreme Universe'' No.~21H05183, and from JSPS KAKENHI, Grant No.~23K03230. R.G. thanks the Council of Scientific and Industrial Research (CSIR), Government of India, for financial support through a fellowship (File No. 09/947(0233)/2019-EMR-I).  K.G. is supported by the Hong
Kong Research Grant Council (RGC) through the grant
No. 17307520, John Templeton Foundation through
grant 62312, “The Quantum Information Structure
of Spacetime” (qiss.fr).  H.B. thanks IIIT Hyderabad for its warm hospitality during his visit in the initial phase of this work. S.D. acknowledges support from the Science and Engineering Research Board, Department of Science and Technology (SERB-DST), Government of India, under Grant No. SRG/2023/000217 and the Ministry of Electronics and Information Technology (MeitY), Government of India, under Grant No. 4(3)/2024-ITEA. S.D. also thanks IIIT Hyderabad for the Faculty Seed Grant.

\medskip
\textit{Authors' contributions.}---F.B., R.G., S.D., and M.N.B. conceived the ideas. F.B., K.G., S.D., and M.N.B. conducted the investigation and developed the methodology, with R.G. also contributing to the methodology development. K.G. contributed to the visualization of the results. The manuscript was written by F.B., R.G., and M.N.B. All authors participated in the discussions, analysis, and validation of results, as well as in writing reports and editing the manuscript during the review process. F.B. and M.N.B. supervised the overall project.

\appendix

\refstepcounter{section}
\section*{APPENDIX A: DETAILS FOR SEC.~\ref{sec:example}}\label{app:revival}
\setcounter{section}{1}
Let us consider the three-time snapshot scenario Eqs.~\eqref{eq:3-point-conf-0}-\eqref{eq:3-point-conf-2}, and consider the initial state of the composite reference-system-environment given by
\begin{align}
&\Phi^+_{RQ}\otimes\Phi^{(4)}_{E\widetilde{E}} \label{eq:example}\\
\xrightarrow{t_1}\ &U_{QE}\;(\Phi^+_{RQ}\otimes\Phi^{(4)}_{E\widetilde{E}})\;U_{QE}^\dag \nonumber\\
&=\frac14\sum_{i,j}\pi^i_Q\Phi^+_{RQ}\pi^j_Q\otimes \ket{i}\bra{j}_E\otimes \ket{i}\bra{j}_{\widetilde{E}} \nonumber\\
\xrightarrow{t_2}\ &U_{QE}\left(\frac14\sum_{i,j}\pi^i_Q\Phi^+_{RQ}\pi^j_Q\otimes \ket{i}\bra{j}_E\otimes \ket{i}\bra{j}_{\widetilde{E}}\right) U_{QE}^\dag \nonumber\\
&=\Phi^+_{RQ}\otimes\Phi^{(4)}_{E{\widetilde{E}}}\;,\nonumber
\end{align}
where $\Phi^{(4)}_{E\widetilde{E}}$ denotes the four-dimensional maximally entangled state purifying $\openone_E/4$, $\Phi^+_{RQ}$ is maximally entangled two qubit state and $U_{QE}=\sum_{i=0}^3\pi^i_Q\otimes\ketbra{i}_E$ is a control unitary operation where $\pi^1_Q=X_Q$, $\pi^2_Q=Y_Q$, and $\pi^3_Q=Z_Q$ are the Pauli matrices and $\pi^0_Q=\openone_Q$.

The crucial point to emphasize here is that the ancilla $\widetilde{E}$, although correlated with $E$, \textit{never directly} interacts with $Q$. For this reason, there cannot be \textit{any} direct backflow of information from $\widetilde{E}$ to $Q$, which would require that some information flowed from $Q$ into $\widetilde{E}$ in the first place. In this case, nonetheless, the ancilla $\widetilde{E}$ alone is able explain the revival, because $I(R;Q\widetilde{E})_0=I(R;Q'\widetilde{E})_1=I(R;Q''\widetilde{E})_2=2$ bits. This can be easily proven since $I(R;Q'\widetilde{E})_1=I(R;\widetilde{E})_1+I(R;Q'|\widetilde{E})_1=0+2=2$ bits. Note that the same conclusion holds if, instead of maximally entangled, $\widetilde{E}$ is just perfectly (classically) correlated with $E$.

\refstepcounter{section}
\section*{APPENDIX B: FACTS ABOUT SQUASHED AND PUFFED CORRELATION MEASURES}\label{App:squashed-puffed}
\setcounter{section}{2}
The \textit{squashed entanglement} is defined as~\cite{christandl-winter-2004-squashed}
\begin{align}
    E_{\sq}(A;B)_{\rho}\coloneqq \dfrac{1}{2} \inf_{\omega_{ABE}\coloneqq\Tr_E{\omega_{ABE}}=\rho_{AB}} I(A;B|E)_{\omega}\;,
\end{align}
where the optimization is done over all possible state extensions $\omega_{ABE}$ of $\rho_{AB}$ and, in general, the dimension of the extension $E$ is unbounded~\cite{christandl-winter-2004-squashed}. Squashed entanglement is an entanglement monotone, i.e., nonincreasing under the action of LOCC (local quantum operations and classical communication) channels. For finite-dimensional bipartite systems, the squashed entanglement is zero if and only if the state is separable~\cite{christandl-winter-2004-squashed,faithful-squashed-entanglement}, i.e., $E_{\rm sq}(A;B)_{\rho}=0$ if and only if the state $\rho_{AB}$ can be expressed as a convex mixture $\sum_xp(x)\sigma^x_A\otimes\omega^x_B$ of product states. 

Given a tripartite state $\rho_{ABC}$, the \textit{squashed non-Markovianity} between $A$ and $C$ conditioned on $B$ is defined as~\cite{Gangwar2024},
\begin{equation}
    N_{\rm sq}(A;C|B)_{\rho}\coloneqq\frac{1}{2}\inf_{\omega_{ABCE}:\Tr_E\omega_{ABCE}=\rho_{ABC}}I(A;C|BE)_{\omega}\;,\label{eq:sqnm}
\end{equation}
where optimization is done with respect to all possible state extensions $\omega_{ABCE}$ of $\rho_{ABC}$. By definition, it thus follows that, for an arbitrary tripartite state $\rho_{ABC}$,
\begin{equation}
    0\leq E_{\rm sq}(A;C)_{\rho}\le N_{\rm sq}(A:C|B)_{\rho}\;.\label{eq:sq-nm}
\end{equation}
Therefore, $N_{\rm sq}(A;C|B)_{\rho}=0$ implies that $E_{\rm sq}(A;C)_{\rho}=0$, i.e., it implies that $\rho_{AC}$ is separable.

The \textit{puffed entanglement} is defined as~\cite{oppenheim2008paradigm}
\begin{align}
    E_{\mathrm{puff}}(A;B)_{\rho}\coloneqq \dfrac{1}{2} \sup_{\omega_{ABE}:\Tr_E{\omega_{ABE}}=\rho_{AB}} I(A;B|E)_{\omega}\;,
\end{align}
where the optimization is done over all possible state extensions $\omega_{ABE}$ of $\rho_{AB}$. It is not an entanglement monotone. For an arbitrary state $\rho_{AB}$, we have
\begin{equation}\label{eq:puffed-bound}
     E_{\mathrm{puff}}(A;B)_{\rho}\leq \min\{H(A)_{\rho},H(B)_{\rho}\}\;.
\end{equation}
The above relation can be easily proved as follows~\cite{private-decoupling,buscemi-AQIS18}. For any given state extension $\omega_{ABE}$ of $\omega_{AB}$, let $\Psi_{ABEF}$ be a purification of $\omega_{ABE}$, i.e., $\Psi_{ABEF}$ is a pure state such that $\Tr_F\Psi_{ABEF}=\omega_{ABE}$. Then, by straightforward manipulation of entropies, we see that
\begin{align*}
    \frac12 I(A;B|E)_{\omega}=H(A)-\frac{I(A;E)_\Psi+I(A;F)_\Psi}{2}\le H(A)\;.
\end{align*}
By symmetry, also $\frac12 I(A;B|E)_{\omega}\le H(B)$ holds.

\refstepcounter{section}
\section*{APPENDIX C:CONVEXITY OF NONCAUSAL REVIVALS}\label{App:convexity}
\setcounter{section}{3}
Let us consider two processes: process (a)
\begin{align}
    \Phi^+_{RQ}\xrightarrow{t_1} \sigma_{RQ'}^{(a)}\xrightarrow{t_2}\tau^{(a)}_{RQ''}\;,
\end{align}
and process (b)
\begin{align}
\Phi^+_{RQ}\xrightarrow{t_1} \sigma_{RQ'}^{(b)}\xrightarrow{t_2}\tau^{(b)}_{RQ''}\;,
\end{align}
and let us begin by assuming that neither of the processes (a) and (b) exhibits a revival, i.e., both $I(R;Q'')_2^{(a)} \leq I(R;Q')_1^{(a)}$ and $I(R;Q'')_2^{(b)} \leq I(R;Q')_1^{(b)}$ hold. Now we construct a probabilistic (convex) mixture of these two processes, namely
\begin{align}\label{eq:convex-process}
\Phi^+_{RQ}\xrightarrow{t_1}\sum_{x\in\{a,b\}} \ p_{x} \sigma_{RQ'}^{(x)}\xrightarrow{t_2}\ \sum_{x\in\{a,b\}} p_x\tau^{(x)}_{RQ''}\;,    
\end{align}
where $p_x\in(0,1)$ are the probabilities and $\sum_{x\in\{a,b\}} \ p_x=1$.

As it is well known, even if the two processes do not display any information revival, their convex mixture in~\eqref{eq:convex-process} may very well do so~\cite{Breuer_2018-mixing-induced}. This means, we may have $I(R;Q'')_2 > I(R;Q')_1$. Indeed, a convex mixture constitutes a process with long-term classical memory: at the beginning, the process is randomly chosen, and the system keeps evolving along that branch until the end. The example described in Section~\ref{sec:example} has exactly this property.

Let us now introduce, for each process, a compatible interaction model; for example
\begin{align}
    \Phi^+_{RQ}\otimes\gamma_E^{(a)}\xrightarrow{t_1} \sigma_{RQ'E'}^{(a)}\xrightarrow{t_2}\tau^{(a)}_{RQ''E''}
\end{align}
and
\begin{align}
    \Phi^+_{RQ}\otimes\gamma_E^{(b)}\xrightarrow{t_1} \sigma_{RQ'E'}^{(b)}\xrightarrow{t_2}\tau^{(b)}_{RQ''E''}\;.
\end{align}

Using the above interaction models we can also construct an interaction model for the mixed process~\eqref{eq:convex-process} as follows
\begin{align}
  &\Phi^+_{RQ}\otimes\sum_x \ p_x\gamma_E^{(x)}\otimes\ketbra{x}_{\widetilde{E}}\\
  \xrightarrow{t_1}&\sum_x \ p_{x} \sigma_{RQ'E'}^{(x)}\otimes\ketbra{x}_{\widetilde{E}}\\
  \xrightarrow{t_2}&\sum_x p_x\tau^{(x)}_{RQ''E''}\otimes\ketbra{x}_{\widetilde{E}}\;,  
\end{align}
where \{$\ket{x}_{\widetilde{E}}\}$ are the orthonormal states of system $\widetilde{E}$, which is also part of the environment. The joint system--interaction is given by the controlled unitary operators
\begin{align}
    U_{QE\widetilde{E}}= \sum_xU_{QE}^{(x)}\otimes\ketbra{x}_{\widetilde{E}}
\end{align}
for the evolution $t_0 \to t_1$, and 
\begin{align}
    V_{Q'E'\widetilde{E}}= \sum_xV_{Q'E'}^{(x)}\otimes\ketbra{x}_{\widetilde{E}}\;
\end{align}
for the evolution $t_1 \to t_2$.

Since we have obtained a controlled-unitary interaction model with the control as a classical memory, it is now easy to find the inert extension showing that, whatever the information revival occurring in~\eqref{eq:convex-process}, it is non-causal. Such an inert extension simply is another system $\widetilde{F}$, perfectly correlated with $\widetilde{E}$, following the intuition of Section~\ref{sec:example}. Indeed, the intermediate state (at $t_1$) for the extended mixed process becomes
\begin{align}
    \sum_x \ p_{x} \sigma_{RQ'E'}^{(x)}\otimes\ketbra{x}_{\widetilde{E}}\otimes\ketbra{x}_{\widetilde{F}}\;,
\end{align}
where \{$\ket{x}_{\widetilde{F}}\}$ are the orthonormal states of $\widetilde{F}$. We let it evolve to $t_2$. 

Clearly, the system $\widetilde{F}$ is not involved in the evolution and never interacts with $Q$: it is inert. And yet, we have
\begin{align}
    &I(R;Q'\widetilde{F})_1-I(R;Q''\widetilde{F})_2\\
    =&I(R;Q'|\widetilde{F})_1-I(R;Q''|\widetilde{F})_2\\
    =&\sum_xp_x[I(R;Q')^{(x)}_1-I(R;Q'')^{(x)}_2]\\
    \ge& 0\;. \label{eq:ExtDPI}
\end{align}
Here, in the first equality, we used the fact that the $I(R;\widetilde{F})_1=I(R;\widetilde{F})_2=0$, as the systems $R$ and $\widetilde{F}$ were uncorrelated initially (at $t_0$) and never interacted thereafter.
The second equality is the consequence of the classical state of the conditioning system $\widetilde{F}$, for instance, $I(R;Q'|\widetilde{F})_1 = \sum_x p_x I(R;Q')^{(x)}_1$. The third inequality is the result of the initial assumption that $I(R;Q')^{(x)}_1-I(R;Q'')^{(x)}_2 \geq 0$, for $x\in\{a,b\}$. 

Along the same lines, we can also show that if process (a) and process (b) have non-causal information revivals, their convex mixtures remain non-causal. In this case, let $F$ be the inert extension that explains the revivals, so that the situation for process (a) is
\[
\Phi^+_{RQ}\otimes\gamma_{EF}^{(a)}\xrightarrow{t_1} \sigma_{RQ'E'F}^{(a)}\xrightarrow{t_2} \tau^{(a)}_{RQ''E''F}\;,
\]
and for process (b)
\[
\Phi^+_{RQ}\otimes\gamma_{EF}^{(b)}\xrightarrow{t_1} \sigma_{RQ'E'F}^{(b)}\xrightarrow{t_2} \tau^{(b)}_{RQ''E''F}\;.
\]
Note that the inert extension $F$ can be taken to be the same \textit{system} (just take it to be the largest one of the two), although, of course, its \textit{state} may differ, i.e., $\gamma_F^{(a)}\neq \gamma_F^{(b)}$.
By definition of non-causal revivals, we have $I(R;Q'F)^{(a)}_1\ge I(R;Q''F)^{(a)}_2$ and $I(R;Q'F)^{(b)}_1\ge I(R;Q''F)^{(b)}_2$.

Again, let $\widetilde{E}$ be a classical control and $\widetilde{F}$ a perfectly correlated copy, the latter remaining completely inert through the whole process. Then, also in this case, we have:
\begin{align}
    &I(R;Q'F\widetilde{F})_1-I(R;Q''F\widetilde{F})_2\\
    =&I(R;Q'F|\widetilde{F})_1-I(R;Q''F|\widetilde{F})_2\\
    =&\sum_xp_x[I(R;Q'F)^{(x)}_1-I(R;Q''F)^{(x)}_2]\\
    \ge& 0\;.
\end{align}

We have thus shown that any convex mixture of processes without revivals or with revivals that are non-causal cannot have any genuine backflow. Therefore, processes without genuine backflow (including processes without revivals and processes with non-causal revivals) form a convex set.

\bibliography{library}   
\end{document}